\documentclass[a4paper,12pt]{article}

\usepackage[ruled]{algorithm2e}
\usepackage{microtype}
\usepackage{amsthm}
\usepackage{graphicx}
\usepackage{color}
\usepackage[all]{xy}
\usepackage{tikz}
\usetikzlibrary{shapes.geometric, arrows.meta, calc, positioning}
\usepackage[11pt]{moresize}

\newtheorem{lemma}{Lemma} 
\newtheorem{theorem}{Theorem}  
\newtheorem{corollary}{Corollary}

\newtheorem{definition}{Definition} 

\theoremstyle{remark}

\def\romannum{\begingroup
  \def\theenumi{\textup{(\roman{enumi})}}%
  \def\p@enumi{}%
  \def\labelenumi{\theenumi}%
  \enumerate}

\newcommand{\LL}[0]{\ensuremath{\mathsf{L}}}
\newcommand{\NP}[0]{\ensuremath{\mathsf{NP}}}

\newcommand{\NL}[0]{\ensuremath{\mathsf{NL}}}

\newcommand{\Ptime}[0]{\ensuremath{\mathsf{P}}}

\newcommand{\HC}[0]{\ensuremath{\mathrm{HC}}}

\newcommand{\DC}[0]{\ensuremath{\mathrm{DC}}}
\newcommand{\DCm}[0]{\ensuremath{\mathrm{DC}^*_m}}
\newcommand{\Cyl}[0]{\ensuremath{\mathrm{Cyl}^*}}
\newcommand{\Cylm}[0]{\ensuremath{\mathrm{Cyl}^*_m}}

\renewcommand{\phi}{\varphi}

\newcommand{\ignore}[1]{}

\newcommand{\G}{\ensuremath{\mathrm{G}}}
\newcommand{\W}{\ensuremath{\mathrm{W}}}
\newcommand{\F}{\ensuremath{\mathrm{F}}}
\renewcommand{\H}{\ensuremath{\mathrm{H}}}
\newcommand{\B}{\ensuremath{\mathrm{B}}}

\begin{document}

 %%%%%%%%%%%%%%%%%%%%%%%%%%%%%%%%%%%%%%%%%%%%%%%%
%\markboth{B. Larose, B. Martin, and D. Paulusma}{Surjective ${\mathbf \H}$-Colouring over Reflexive Digraphs} 
 
\title{Surjective ${\mathbf \H}$-Colouring over Reflexive Digraphs\thanks{An extended abstract of this article has appeared at the 35th International Symposium on Theoretical Aspects of Computer Science (STACS 2018).
}}

%\author[1]{Beno\^it Larose\thanks{Author supported by NSERC and FRQNT.}}
%\author[2]{Barnaby Martin}
%\author[2]{Dani\"el Paulusma\thanks{Author supported by The Leverhulme Trust (RPG-2016-258).}}

\author{
Beno\^it Larose\thanks{The first author was supported by NSERC and FRQNT.} \\
LACIM, Universit\'e du Qu\'ebec a Montr\'eal, Canada
\and
Barnaby Martin \\
Department of Computer Science, Durham University, U.K.
\and
Dani\"el Paulusma\thanks{The third author was supported by The Leverhulme Trust (RPG-2016-258).} \\
Department of Computer Science, Durham University, U.K.
}

%\affil[1]{LACIM, Universit\'e du Qu\'ebec a Montr\'eal, Canada}
%\affil[2]{Department of Computer Science, Durham University, U.K.}

%\authorrunning{B. Larose, B. Martin and D. Paulusma}

%\Copyright{B. Larose, B. Martin and D. Paulusma}

\maketitle

\begin{abstract}
The {\sc Surjective $\H$-Colouring} problem is to test if a given graph allows a vertex-surjective homomorphism to a fixed 
graph~$\H$. The complexity of this problem has been well studied for undirected (partially) reflexive graphs.
We introduce
{\it endo-triviality}, the property of a 
structure that all of its endomorphisms that do not have range of size $1$ are automorphisms, 
as a means to obtain complexity-theoretic classifications of {\sc Surjective $\H$-Colouring} in the case of reflexive {\it digraphs}.
Chen [2014] proved, in the setting of constraint satisfaction problems, that {\sc Surjective $\H$-Colouring} is \NP-complete if $\H$ 
has the property that all of its polymorphisms are essentially unary. 
We give the first concrete application of his result by showing that every 
endo-trivial 
reflexive digraph~$\H$ has this property.
We then use the concept of 
endo-triviality
to prove, as our main result, a dichotomy for {\sc Surjective $\H$-Colouring} when $\H$ is a reflexive tournament: if $\H$ is transitive, then {\sc Surjective $\H$-Colouring} is in \NL, otherwise it is \NP-complete.

By combining this result with some known and new results we obtain a complexity classification for {\sc Surjective $\H$-Colouring} when $\H$ is a partially reflexive digraph of size at most $3$.

\end{abstract}
%\category{F.2.2}{Analysis of Algorithms and Problem Complexity}{Nonnumerical Algorithms and Problems}
%\category{G.2.1}{Discrete Mathematics}{Combinatorics}

%\keywords{Surjective H-Coloring, Computational Complexity, Algorithmic Graph Theory, Universal Algebra, Constraint Satisfaction}
%%%%%%%%%%%%%%%%%%%%%%%%%%%%%%%%%%%%%%%%%%%%%%%%%%%%%%%%%

%Editor-only macros:: begin (do not touch as author)%%%%%%%%%%%%%%%%%%%%%%%%%%%%%%%%%%

%An extended abstract of this article has appeared at the 35th International Symposium on Theoretical Aspects of Computer Science (STACS 2018).

%The first author was supported by NSERC and FRQNT and the third author was supported by The Leverhulme Trust (RPG-2016-258).

\section{Introduction}\label{s-intro}

The classical \emph{homomorphism problem}, also known as {\sc $\H$-Colouring},
involves a fixed structure~$\H$, with input another structure~$\G$, of the same signature, invoking the question as to whether there is a function from the domain of~$\G$ to the domain of $\H$ that is a homomorphism from $\G$ to $\H$. 
The {\sc $\H$-Colouring} problem is an intensively studied problem, which has additionally attracted attention in its guise of the \emph{constraint satisfaction problem} (CSP), especially since the seminal paper of Feder and
Vardi~\cite{FederVardi}. 
Their well-known conjecture, recently proved by Bulatov~\cite{FVproofBulatov} and Zhuk~\cite{FVproofZhuk}, stated that every CSP$(\H)$ has complexity either in P or \NP-complete, omitting any Ladner-like complexities in between. 

This paper concerns the computational complexity of 
the \emph{surjective homomorphism} problem, also known in the literature as {\sc Surjective $\H$-Colouring} \cite{CiE2017,GolovachPS12} and {\sc $\H$-Vertex-Compaction} \cite{VikasAlgorithmica}. This problem requires the homomorphism
to be surjective.
It is a cousin of the \emph{list homomorphism} problem and is even more closely related to the \emph{retraction} and \emph{compaction} problems. Indeed, the {\sc $\H$-Compaction} problem, hitherto defined only for graphs~$\H$, takes as input a graph $\G$ and asks if there exists a function $f$ from $V(\G)$ to $V(\H)$ so that for each non-loop edge $(x,y) \in E(\H)$ (\mbox{i.e.} with $x\neq y$), there exists $u,v \in V(\G)$ so that $f(u)=x$ and $f(v)=y$. Thus, compaction can be seen as the \emph{edge-surjective homomorphism} problem.\footnote{Except for the treatment of self-loops, which appears to be an idiosyncrasy that plays no vital role in computational complexity. For some history of the definition see \cite{VikasReflexive}.} The problem {\sc $\H$-Retraction} takes as input a superstructure~$\G$ of $\H$ and asks whether there is a homomorphism from $\G$ to $\H$ that is the identity on $\H$. 
The {\sc $\H$-Retraction} problem is polynomially equivalent with a special type of CSP, CSP$(\H')$, where $\H'$ is $\H$ decorated with constants naming the elements of its domain.
 Feder and Vardi~\cite{FederVardi} showed that
the task of classifying the complexities of the retraction problems is equivalent to that for the CSPs. 
Hence, owing to ~\cite{FVproofBulatov,FVproofZhuk}, {\sc $\H$-Retraction} has now been fully classified.

The list homomorphism problem, {\sc List $\H$-colouring}, allows one to express restricted lists for each of the input structure's elements, that are the only domain elements permitted in a solution homomorphism. {\sc List $\H$-colouring} is also a special type of CSP, CSP$(\H')$, where $\H'$ is $\H$ replete with all possible unary relations over the domain of $\H$. Historically, the complexities of {\sc List $\H$-colouring} were the first to be settled by Bulatov~\cite{ConservativeJournal}, following important earlier work on graphs \cite{FederHell98,FederHellHuang99,FHH03}. 
%For a thorough treatment of graphs and homomorphisms we refer to the book \cite{HNBook}. 

In contrast to the situation for {\sc $\H$-Colouring}, {\sc List $\H$-Colouring} and
{\sc $\H$-Retraction},
 the complexity classifications for {\sc $\H$-Compaction} and {\sc Surjective $\H$-Colouring} are far from settled, and there are concrete open cases (see {\sc $3$-No-Rainbow-Colouring} in the survey~\cite{SurHomSurvey}). Obtaining \NP-hardness for compaction and surjective homomorphism problems appears to be especially challenging. The complexity-theoretic relationship between these various problems is drawn in Figure~\ref{fig:Anthony}.
At present it is not known whether there is a graph $\H$ so that {\sc $\H$-Retraction},  {\sc $\H$-Compaction} and {\sc Surjective $\H$-Colouring} do not have the same complexity up to polynomial time reduction (see~\cite{CiE2017,VikasOther}).

\begin{figure}
\begin{centering}
\newcommand{\separation}{3mm}
\begin{tikzpicture}[problem/.style={draw,rectangle,rounded corners,minimum width=1cm,minimum height=0.5cm,font=\ssmall}]
\node[problem] (listHcolouring) {{\sc List $\H$-Colouring}};
\node[problem, right=\separation of listHcolouring] (Hretraction)  {{\sc $\H$-Retraction}};
\node[problem, right=\separation of Hretraction] (Hcompaction)  {{\sc $\H$-Compaction}};
\node[problem, right=\separation of Hcompaction] (surjectiveHcolouring) {{\sc Surj $\H$-Colouring}};
\node[problem, right=\separation of surjectiveHcolouring] (Hcolouring)  {{\sc $\H$-Colouring}};
\draw[->, thick] (listHcolouring) to (Hretraction);
\draw[->, thick] (Hretraction) to (Hcompaction);
\draw[->, thick] (Hcompaction) to (surjectiveHcolouring);
\draw[->, thick] (surjectiveHcolouring) to (Hcolouring);
\end{tikzpicture}
\caption{Relations between {\sc Surjective $\H$-Colouring} and its variants (from \cite{CiE2017}). An arrow from one problem to another indicates that the latter problem is polynomial-time solvable for a graph~$\H$ whenever the former is polynomial-time solvable for $\H$. Reverse arrows do not hold for the leftmost and rightmost arrows,
as witnessed by the reflexive 4-vertex cycle for the rightmost arrow and by any reflexive tree that is not a reflexive interval graph for the leftmost arrow (Feder, Hell and Huang~\cite{FHH03} showed that the only reflexive bi-arc graphs are reflexive interval graphs).
It is not known whether the reverse direction holds for the two middle arrows.}
\label{fig:Anthony}
\end{centering}	
\end{figure}
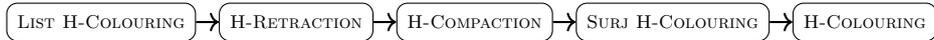

Nevertheless classification results for {\sc Surjective $\H$-Colouring} have tried to keep pace with similar ones for  {\sc $\H$-Retraction}. In \cite{pseudoforests} it is proved, among partially reflexive pseudoforests~$\H$, where the problem {\sc $\H$-Retraction} splits between P and \NP-complete. A similar classification for {\sc Surjective $\H$-Colouring} over partially reflexive forests can be inferred from the classification for partially reflexive trees in \cite{GolovachPS12}. The quest for a classification for {\sc $\H$-Compaction} and {\sc Surjective $\H$-Colouring} over pseudoforests is ongoing, but for both problems already the reflexive $4$-cycle took some time to classify \cite{JCTBMartinPaulusma,VikasReflexive}, as well as the irreflexive $6$-cycle~\cite{VikasIrreflexive,VikasMFCS}. 

The above results are for undirected graphs, whereas we focus on digraphs.
A known classification for {\sc $\H$-Retraction} comes for irreflexive semicomplete digraphs~$\H$. In \cite{Semicomplete}
Bang-Jensen, Hell, and MacGillivray proved that {\sc $\H$-Colouring} is always in \Ptime\ or is \NP-complete if $\H$ is irreflexive semicomplete. This is a fortiori a classification for {\sc $\H$-Retraction} since semicomplete digraphs are \emph{cores} (all endomorphisms are automorphisms), which ensures that {\sc $\H$-Colouring} and {\sc $\H$-Retraction} are polynomially equivalent. 
%Their paper inspired the work  of Dapic, Markovic, and Martin~\cite{SemicompleteQCSP}, where the object of study remained semicomplete~$\H$, but the problem became the \emph{Quantified CSP} and the method moved from largely combinatorial to largely algebraic (admittedly, with heavy combinatorics within).  
For irreflexive semicomplete digraphs $\H$, the classification for {\sc Surjective $\H$-Colouring} can be read trivially from that for {\sc $\H$-Colouring}, and they are the same. An obvious next place to look 
is at
the situation if $\H$ is {\it reflexive} semicomplete, where surely the classifications will not be the same as {\sc $\H$-Colouring} is trivial in this case. 

Reflexive tournaments form an important subclass of the class of reflexive semicomplete graphs
and are well-understood algebraically~\cite{larose2006taylor}.  
In particular,
the classification for {\sc $\H$-Retraction} where $\H$ is a reflexive tournament can be inferred from 
the algebraic characterisation from~\cite{larose2006taylor}: for a reflexive tournament~$\H$, the {\sc $\H$-Retraction} problem is in NL if $\H$ is transitive, and it is \NP-complete otherwise. This raises the question whether the same holds for {\sc Surjective $\H$-Colouring} and whether we can develop algebraic methods further to prove this.
In fact, the algebraic method is by now well known for CSPs and their relatives, including its use with digraphs; see the recent survey \cite{LaroseSurvey17}. However, the algebraic method is not so far advanced for surjective homomorphism problems.
So far it only exists in the work of Chen~\cite{ChenSurHom}, who proved that {\sc Surjective $\H$-Colouring} is \NP-complete if $\H$ 
has the property that all of its polymorphisms depend only on one variable, that is, are essentially unary. 
Chen's result has not yet been put to work (even on toy open problems) and a key driver for our research has been to find, in the wild, a place for its application.

\medskip
\noindent
{\bf Our Results.}
We give, for the first time, complexity classifications for {\sc Surjective $\H$-Colouring} for digraphs instead of undirected graphs.
To prove our results, we further develop algebraic machinery to tackle surjective homomorphism problems. That is, in Section~\ref{s-pre} we introduce, after giving the necessary terminology, the concept of 
endo-triviality. We show how this concept is closely related to some known algebraic concepts and
explore its algorithmic consequences in the remainder of our paper.

Firstly, in Section~\ref{sec:Benoit}, we prove that a reflexive digraph $\H$ that is 
endo-trivial
has the property that all of its polymorphisms are essentially unary. Combining this result with the aforementioned result of Chen~\cite{ChenSurHom} immediately yields that
{\sc Surjective $\H$-Colouring} is \NP-complete for any such digraph~$\H$. This is 
the first concrete application of Chen's result to settle a problem of open complexity;
it shows, for instance, that {\sc Surjective $\H$-Colouring} is \NP-complete if $\H$ is a reflexive directed cycle on $k\geq 3$ vertices.
As the case $k\leq 2$ is trivial, this gives a classification of {\sc Surjective $\H$-Colouring} for reflexive directed cycles, which we believe form  a natural class of digraphs to consider given the results in~\cite{JCTBMartinPaulusma,VikasMFCS}.

Secondly, in Section~\ref{s-tour} we give a complexity classification for {\sc Surjective $\H$-Colouring}, when $\H$ is a reflexive tournament.
We use endo-triviality in an elaborate and recursive encoding of an \NP-hard retraction problem within {\sc Surjective $\H$-Colouring}.  In doing this, we show that on this class, the complexities of {\sc Surjective $\H$-Colouring} and {\sc \H-Retraction} coincide.

Finally, 
our results enable us 
to give a complexity classification for {\sc Surjective $\H$-Colouring} when $\H$ is a partially reflexive digraph of size at most $3$.  In doing this, we show that on this class, the complexities of {\sc Surjective $\H$-Colouring} and {\sc \H-Retraction} coincide. We are not aware of an existing classification for {\sc \H-Retraction} on this class, but we do build on one existing for {\sc List \H-Colouring} from \cite{FederHellTucker}. %This section is deferred entirely to the appendix, as are many of our proofs, for reasons of space.

\section{Preliminaries}\label{s-pre}

Let $[n]:=\{1,\ldots,n\}$. For a $k$-tuple $\overline{t}$ and $i \in [k]$, let $\overline{t}[i]$ be the $i$th entry in $\overline{t}$. In a digraph $\G$, a forward- (resp., backward-) \emph{neighbour} (or \emph{adjacent}) to a vertex $u \in V(\G)$ is another vertex $v \in V(\G)$ so that $(u,v) \in E(\G)$ (resp., $(v,u) \in E(\G)$). The \emph{out-degree} and \emph{in-degree} of a vertex are the number of its forward-neighbours and backward-neighbours, respectively. 
A vertex with out-degree and in-degree both 0 is said to be \emph{isolated}. A vertex with a self-loop is \emph{reflexive} and otherwise it is \emph{irreflexive}. A digraph is \emph{(ir)reflexive} if all its vertices are (ir)reflexive.

The \emph{directed path} on $k$ vertices is the digraph with vertices $u_0,\ldots, u_{k-1}$ and edges $(u_i,u_{i+1})$ for $i=0,\ldots,k-2$.
The \emph{directed cycle} on $k$ vertices is obtained from the directed path on $k$ vertices after adding the edge $(u_{k-1},u_0)$.
A digraph $\G$ is \emph{strongly connected} if for all $u,v \in V(\G)$ there is a 
directed path 
in $E(\G)$ from $u$ to $v$ (note that we take this to include the situation $u=v$, but for reflexive graphs the distinction is moot). A digraph is \emph{weakly connected} if its symmetric closure (underlying undirected graph) is connected. A \emph{double-edge} in a digraph $\G$ consists in a pair of distinct vertices $u,v \in V(\G)$, so that $(u,v),(v,u) \in E(\G)$. 
A digraph $\G$ is \emph{semicomplete} if for every two distinct vertices $u$ and $v$, at least one of $(u,v)$, $(v,u)$ belongs to $E(\G)$. 
A digraph~$\G$ is a \emph{tournament} if for every two distinct vertices $u$ and $v$, exactly one of $(u,v)$, $(v,u)$ belongs to $E(\G)$.
We demand our tournaments have more than one vertex (to rule out certain trivial cases in proofs). A reflexive tournament $\G$ is \emph{transitive} if for every triple of vertices $u,v,w$ with $(u,v), (v,w)\in E(\G)$, also $(u,w)$ belongs to $E(\G)$.
 A digraph~$\F$ is a \emph{subgraph} of a digraph $\G$ if $V(\F) \subseteq V(\G)$ and $E(\F) \subseteq E(\G)$. It is \emph{induced} if $E(\F)$ coincides with $E(\G)$ restricted to pairs containing only vertices of $V(\F)$. A \emph{subtournament} is an induced subgraph of a tournament (note that this is a fortiori a tournament). All subgraphs we consider in this paper will be induced.

A \emph{homomorphism} from a digraph~$\G$ to a digraph~$\H$ is a function $f:V(\G)\rightarrow V(\H)$ so that for all $u,v \in V(\G)$ with $(u,v) \in E(\G)$ we have $(f(u),f(v)) \in E(\H)$. We say that $f$ is \emph{(vertex)-surjective} if for every vertex $x\in V(\H)$ there exists a vertex $u\in V(\G)$ with $f(u)=x$.
Let $\H$ be a digraph. A \emph{homomorphic image} of $\H$ is a digraph $\H'$ so that there is a surjective homomorphism $h:\H \rightarrow \H'$ in which, for all $(x',y') \in E(\H')$ there exists $(x,y) \in E(\H)$ so that $x'=h(x)$ and $y'=h(y)$. That is, $h$ is vertex- and edge-surjective.  

The \emph{direct product} of two digraphs $\G$ and $\H$, denoted $\G \times \H$, has vertex set $V(\G) \times V(\H)$ and edges $((x,y),(x',y'))$ exactly when $(x,x') \in E(\G)$ and $(y,y') \in E(\H)$. 
This product is associative and commutative, up to isomorphism, and spawns a natural power.
A $k$-ary \emph{polymorphism} of $\G$ is a function $f:\G^k\rightarrow \G$ so that when $(x_1,y_1),\ldots,$ $(x_k,y_k) \in E(\G)$ then $(f(x_1,\ldots,x_k),f(y_1,\ldots,y_k)) \in E(\G)$. A polymorphism of $\G$ can be seen as a homomorphism from the $k$th (direct) power of $\G$, $\G^k$, to $\G$. A polymorphism $f$ is \emph{idempotent} if for all $x\in V(\G)$, $f(x,\ldots,x)=x$.
The $k$-ary $i$th \emph{projection}, for $i \in [k]$, is the polymorphism $\pi^i_k$ given by $\pi^i_k(x_1,\ldots,x_k)=x_i$. 
A $k$-ary operation $f$ is called \emph{essentially unary} if there exists a unary operation $g$ and $i \in [k]$ so that $f(x_1,\ldots,x_k)=g(x_i)$
for all $(x_1,\ldots,x_k)\in \G^k$.

 Let $\G$ be a digraph. An \emph{endomorphism} of~$\G$ is a homomorphism from $\G$ to itself. 
 An endomorphism~$e$ of $\G$ is a \emph{constant map} if there exists a vertex $v\in V(\G)$ such that $e(u)=v$ for all $u\in V(\G)$.
The \emph{self-map digraph} $\G^{\G}$ has as its vertices the self-maps of $V(\G)$, and there is an edge $(f,g) \in E(\G^\G)$ between self-maps $f$ and $g$ if and only if for every edge $(x,y) \in E(\G)$, we have that $(f(x),g(y)) \in E(\G)$. The \emph{endomorphism digraph} $\widehat{\G^{\G}}$ is the restriction of the self-map digraph $\G^{\G}$ to the vertices induced by endomorphisms of $\G$. Note that the self-loops of $\G^{\G}$ are precisely the endomorphisms of $\G$, so $\widehat{\G^{\G}}$ is reflexive when $\G$ is reflexive. We now make two more observations. 
The first one follows directly from the definition of $\widehat{\G^{\G}}$ as well. 
The second one can, for example, be found in Section~5.2 of~\cite{llt}.

\begin{lemma}\label{l-gg2}
If $(f_1,g_1) \in E(\widehat{\G^\G})$ and $(f_2,g_2) \in E(\widehat{\G^\G})$, then $((f_1 \circ f_2),(g_1 \circ g_2)) \in E(\widehat{\G^\G})$. 
\end{lemma}

\begin{lemma}\label{l-gg3}
Let $\G$ and $\H$ be two digraphs. Let $\phi$ be a homomorphism from $\H\times \G$ to~$\G$.
Then the function $\psi$ defined by $\psi(x)(u)=\phi(x,u)$ for all $x\in V(\H)$, $u\in V(\G)$ 
is a homomorphism from $\H$ to $\widehat{\G^{\G}}$.
\end{lemma} 
 
A bijective endomorphism whose inverse is a homomorphism is an \emph{automorphism}.
An endomorphism is \emph{non-trivial} if it is neither an automorphism nor a constant map.
 A digraph, all of whose endomorphisms are automorphisms, is termed a \emph{core}. 
An endomorphism $e$ of a digraph~$\H$ \emph{fixes} a subset $S \subseteq V(\H)$ if $e(S)=S$, that is, $e(x)\in S$ for all $x \in S$, 
and it fixes a subgraph~$\F$ of $\H$ if $e(\F)=\F$. It fixes an induced subgraph~$\F$ \emph{up to automorphism} if $e(\F)$ is an automorphic copy of $\F$ (this is a stronger condition than $e(\F)$ being isomorphic to $\F$). 
An endomorphism~$r$ of $\G$ is a \emph{retraction} of $\G$ if $r$ is the identity on the image $r(\G)$ (thus a retraction must have at least one fixed point).

\begin{figure}
\centering
\input{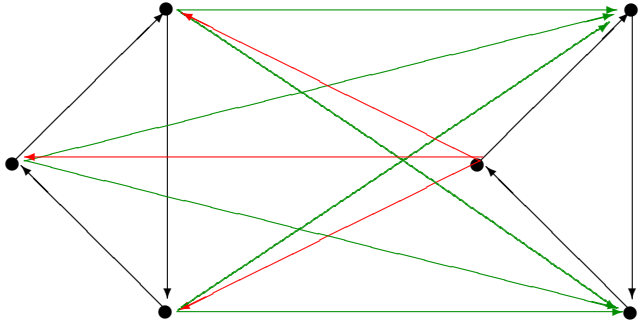}
\caption{A tournament on six vertices (self-loops are not drawn), which retracts to the directed 3-cycle (in black) on the right-hand side, but not to the one on the left-hand side (in black as well). However, there is no endomorphism that maps the left-hand one isomorphically to the right. We can use this tournament to build a structure that is a counterexample to the generalisation of Lemma~\ref{lem:endo-retraction-trivial} stating that endo-trivial and retract-trivial coincide. Let us label the vertices in the tournament: $\alpha,\beta,\gamma$ (left-hand $\DC^*_3$, clockwise from bottom) and $0,1,2$ (right-hand $\DC^*_3$, clockwise from bottom). Let us build a structure $\B$ by augmenting a new $6$-ary relation with tuples in $\{(\alpha,\beta,\gamma,0,1,2),(\alpha,\alpha,\alpha,\alpha,\beta,\gamma),(\alpha,\alpha,\alpha,\alpha,\alpha,\alpha)\}$. 
The structure~$\B$ is retract-trivial but is not endo-trivial, since it has an interesting endomorphism that takes $(\alpha,\beta,\gamma,0,1,2)$ to $(\alpha,\alpha,\alpha,\alpha,\beta,\gamma)$.}
\label{fig:Tournament6}
\end{figure}

\medskip
\noindent
{\bf Endo-triviality and Retract-triviality.}
We now define the key concept of endo-triviality and the closely related concept of retract-triviality.

\begin{definition}
A digraph is \emph{endo-trivial} if all of its endomorphisms are automorphisms or constant maps. 
\end{definition}

\noindent The concept of endo-triviality also arises from the perspective of the algebra of polymorphisms. An algebra is called minimal if its unary polynomials are either constants or the permutations (see Definition~2.14 in \cite{HobbyMcKenzie}). For reflexive digraphs,
polynomials and polymorphisms coincide. 
In other words, a reflexive digraph is endo-trivial if and only if its associated algebra of polymorphisms is minimal.

We will also need the following closely related concept.

\begin{definition}
A digraph is \emph{retract-trivial} if all of its retractions are the identity or constant maps.
\end{definition}

\noindent The concept of retract-triviality also appears in the algebraic theory but has, as far as we are aware, not been studied in a combinatorial setting. An algebra is term-minimal if the only retractions in its clone of terms are the identity and constants (see~\cite{TermMinimalAlgebras}). A reflexive digraph is retract-minimal if its associated algebra of polymorphisms is term-minimal. It follows that on reflexive digraphs, the concepts of retract-minimality and retract-triviality coincide.

We note that every endo-trivial structure is also retract-trivial. However, the reverse implication is not necessarily true: 
in Figure~\ref{fig:Tournament6} we give an example of a structure that is retract-trivial but not endo-trivial.
This example is based on a digraph but is not itself a digraph. It is also possible to construct a retract-trivial digraph that is not endo-trivial~\cite{Si17}, but
on reflexive tournaments both concepts do coincide.

\begin{lemma}
A reflexive tournament is endo-trivial if and only if it is retract-trivial.
\label{lem:endo-retraction-trivial}
\end{lemma}

\begin{proof}
(Forwards.) Trivial.
(Backwards.) By contraposition, suppose $e$ is a non-trivial endomorphism of a reflexive tournament $\H$. Consider $e(\H)$ and build some function $e^{-1}$ from $e(\H)$ to $\H$ by choosing $e^{-1}(y)=x$ if $e(x)=y$ arbitrarily. Since $\H$ is a (reflexive) tournament, $e^{-1}$ is an isomorphism, whereupon $e^{-1} \circ e$ is the identity automorphism when restricted to some subtournament $\H_0$ of $\H$. Hence $e^{-1} \circ e$ is a non-trivial retraction of $\H$ (to $\H_0$).
\end{proof}

\section{Essential Unarity and a Dichotomy for Reflexive Directed Cycles}\label{sec:Benoit}

In this section we give the first concrete application, of which we are aware, of the aforementioned result of Chen, formally stated below.

\begin{theorem}[Corollary 3.5 in~\cite{ChenSurHom}]
Let $\H$ be a finite structure whose universe $V(\H)$ has size strictly greater than~$1$. If each
polymorphism of $\H$ is essentially unary, then {\sc Surjective $\H$-Colouring}
is \NP-complete.
\label{thm:Chen-hauptsatz}
\end{theorem}

In order to this, we make use of the endomorphism graph and a result from M\'aroti and Z\'adori~\cite{marzad}.
Let $id_\H$ denote the identity map on a digraph~$\H$.

\begin{lemma}[Lemma 2.2 in~\cite{marzad}] Let $\H$ be a reflexive digraph. If $(id_\H,f) \in E(\widehat{\H^\H})$, where $f$ is different from
$id_\H$, then $\H$ has a
non-surjective retraction $r$ 
 such that $(id_\H,r)\in E(\widehat{\H^{\H}})$.
\label{lem:marzad}
\end{lemma}

The following lemma is crucial and will be of use in the next section as well.

\begin{lemma}  Let $\H$ be a retract-trivial reflexive digraph with at least three vertices. Then 
\begin{enumerate}
\item $\H$ has no double edge;
\item $\H$ is strongly connected; and
\item the automorphisms of $\H$ are isolated vertices in  $\widehat{\H^{\H}}$. 
\end{enumerate}
\label{lem:Benoit-lemma}
\end{lemma}
\begin{proof} (1) As any reflexive digraph can be retracted onto a double edge, the result follows. \\[3pt]
(2) A reflexive digraph that is not weakly connected may be retracted onto a 2-vertex digraph. So we may safely assume $\H$ is weakly connected. If~$\H$ is not strongly connected, then $\H$ has an edge $(a,b)$ such that $a$ and $b$ are in different strong components, that is, there is no directed path from $b$ to $a$. We define a retraction of $\H$ onto the subgraph induced by $\{a,b\}$ as follows: let $r(x) = a$ if there exists a directed path from $x$ to $a$ and $r(x)=b$ otherwise. 
As $r(a)=a$ and $r(b)=b$, it remains to check if $r$ is an endomorphism. Let $(x,y)$ be an edge. If $r(y)=a$, then there exists
a directed path from $y$ to~$a$, and thus a directed path from $x$ to~$a$ implying that $r(a)=a$. If $r(y)=b$, then $r(x)=a$ 
and $r(x)=b$ are both allowed.\\[3pt]
(3) We first prove that no constant map is adjacent to an automorphism in $\widehat{\H^{\H}}$.
Suppose for a contradiction that  there exists an automorphism $\sigma$ backwards-adjacent to some constant (the case of forwards-adjacent is dual). Thus without loss of generality we may assume that 
$(\sigma,c) \in E(\widehat{\H^\H})$  for some constant map $c$, say for all $u\in V(\H)$, $c(u)=v$ for some $v\in V(\H)$.
By composing sufficiently many times via Lemma~\ref{l-gg2},  we obtain that $(id_\H,c) \in E(\widehat{\H^\H})$. 
Since $\H$ is strongly connected, every vertex $u\in V(\H)$ has out-degree at least~1.
Hence $(u,v) \in E(\H)$ for all $u \in V(\H)$.  Since $v$ has out-degree at least~1, this means that $\H$ contains a 
double edge, contradicting statement~(1).

Now suppose that there exists an automorphism $\sigma$ that is backward-adjacent to some endomorphism $f \neq \sigma$ (the case of forwards-adjacent is dual). Thus, without loss of generality we have $(\sigma,f) \in E(\widehat{\H^\H})$. Apply $\sigma^{-1}$ to both sides of the edge to obtain that $(id_\H,g) \in E(\widehat{\H^\H})$ for some $g \neq id_\H$.  By the preceding lemma, and the fact that $\H$ is retract-trivial, this means $g$ is a constant map, contradicting the claim above.   \end{proof}

\noindent We use Lemma~\ref{lem:Benoit-lemma} to obtain the following structural result.

\begin{theorem} Let $\H$ be  
an endo-trivial 
reflexive digraph with at least three vertices. Then every polymorphism of  $\H$ is essentially unary.  
\label{cor:Benoit}
\end{theorem}

\begin{proof} 
Since $\H$ is endo-trivial, $\H$ is retract-trivial. Hence, 
by Lemma~\ref{lem:Benoit-lemma}, $\H$ is strongly connected, and furthermore the automorphisms of $\H$ are isolated vertices of~$\widehat{\H^{\H}}$. As $\H$ is endo-trivial, this means that  $\widehat{\H^{\H}}$  is the disjoint union of a copy of $\H$ that corresponds to the constant maps and a set of isolated vertices, one for each automorphism of $\H$.
Suppose for a contradiction that there exists an  an $n$-ary polymorphism~$f$ of $\H$ which is not essentially unary.
We may without loss of generality assume that~$f$ depends on all of its $n$ variables, where $n \geq 2$. 
By Lemma~\ref{l-gg3}, 
 the mapping $F:\H^{n-1} \rightarrow \widehat{\H^{\H}}$ defined by $F(x_1,\ldots,x_{n-1})(y) = f(x_1,\ldots,x_{n-1},y)$ is a homomorphism. Since $\H$ is strongly connected, so is $\H^{n-1}$, and hence so is the image of $F$ in $\widehat{\H^{\H}}$. Thus this image is either contained in the component of constants, in which case $f$ does not depend on its last variable, else it is a singleton, in which case $f$ does not depend on any of its first $n-1$ variables. \end{proof}

Combining Theorems~\ref{thm:Chen-hauptsatz} and~\ref{cor:Benoit} yields the main result of this section.

\begin{corollary}\label{c-main}
If $\H$ is an endo-trivial reflexive digraph on at least three vertices, then {\sc Surjective $\H$-Colouring} is \NP-complete. 
\end{corollary}

Let $\DC^*_k$ denote the reflexive directed cycle on $k$ vertices, which is readily seen to be endo-trivial. Corollary~\ref{c-main}
yields the following dichotomy for reflexive directed cycles after noting that {\sc Surjective $\DC^*_k$-Colouring} is trivial for $k\leq 2$.

\begin{corollary}
{\sc Surjective $\DC^*_k$-Colouring} is in \LL\ if $k\leq 2$ and \NP-complete if $k\geq 3$.
\end{corollary}

It is not difficult to construct endo-trivial reflexive tournaments other than reflexive directed cycles.
In the next section though we give a combinatorial \NP-hardness proof for {\sc Surjective $\H$-Colouring} whenever $\H$ is {\it any} non-transitive reflexive tournament. As $\DC^*_3$ is such a digraph, this proof also can be used for the case $\H=\DC^*_3$.
However, it does not extend to {\sc Surjective $\DC^*_k$-Colouring} for $k\geq 4$.

\section{A Dichotomy for Reflexive Tournaments}\label{s-tour}
\label{sec:main}

In this section we prove our main result, namely a dichotomy of {\sc Surjective $\H$-Colouring} for reflexive tournaments~$\H$ by showing that transitivity is the crucial property for tractability. In the next subsections we prove that {\sc Surjective $\H$-Colouring} is \NP-complete when $\H$ is a non-transitive tournament. 

\subsection{Two Elementary Lemmas}
It is well-known that every strongly connected tournament has a directed Hamilton cycle~\cite{Camion59}.
Hence we derive the following corollary to Lemmas~\ref{lem:endo-retraction-trivial} and~\ref{lem:Benoit-lemma} Part 2.
\begin{lemma}
If $\H$ is a reflexive tournament that is endo-trivial, then $\H$ contains a directed Hamilton cycle.
\label{l-Ham}
\end{lemma}

We will also need the following lemma.

\begin{lemma}
If $\H$ is a reflexive tournament that is endo-trivial, then any homomorphic image of $\H$ of size $1<n<|V(\H)|$ possesses a double edge.
\label{lem:under-m}
\end{lemma}
\begin{proof}
Suppose $\H$ has a homomorphic image of size $1<n<|V(\H)|$ without a double edge. By looking at the equivalence classes of vertices identified in the homomorphic image, we can deduce a non-trivial retraction, namely by mapping each of the vertices in an equivalence class to any particular one of them.
\end{proof}
\begin{figure}
\centering

\input{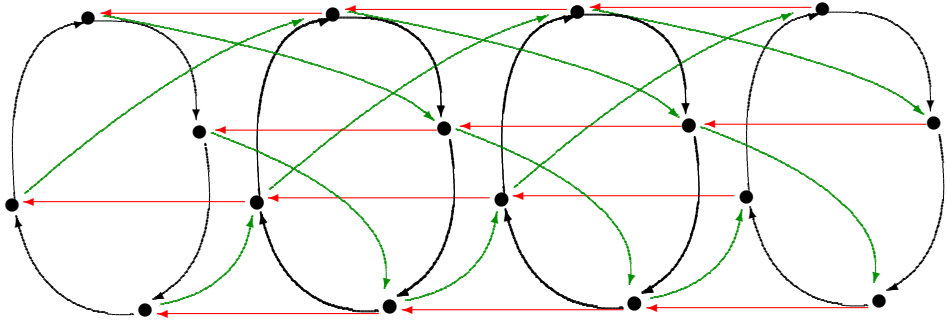}
\caption{The gadget $\Cylm$ in the case $m:=4$ (self-loops are not drawn). We usually visualise the right-hand copy of $\DC^*_4$ as the ``bottom'' copy and then we talk about vertices ``above'' and ``below'' according to the red arrows.}
\label{fig:Photo2}
\end{figure}

\subsection{The NP-Hardness Gadget}
We now introduce the gadget $\Cylm$ drawn in Figure~\ref{fig:Photo2}. We take $m$ disjoint copies of the directed $m$-cycle $\DCm$ arranged in a cylindrical fashion so that there is an edge from $i$ in the $j$th copy to $i$ in the $j+1$th copy (drawn in red), and an edge from $i$ in the $j+1$th copy to $i+1$ in the $j$th copy (drawn in green). We consider $\DCm$ to have vertices $\{1,\ldots,m\}$.
% which we hope will not create too much confusion when we join different copies of these together. 
A key role will be played by Hamilton cycles $\HC_m$ in a strongly connected reflexive tournament on $m$ vertices. We consider this cycle also labelled $\{1,\ldots,m\}$, in order to attach it to the gadget $\Cylm$. 
%This may involve re-labelled vertices that were previously labelled (which might be possible in more than one way up to isomorphism). Again, this should not cause any confusion, and where there is the ambiguity, any choice will work. However, in general, it may be significant which particular Hamilton cycle is chosen.
%
The gadget $\Cylm$ is an alteration of a gadget that appears in \cite{FederHell98} for proving that {\sc List $\H$-Colouring} is \NP-complete when $\H$ is an undirected cycle on at least four vertices, but our proof is very different.

The following lemma follows from induction on the copies of $\DCm$, since a reflexive tournament has no double edges.
\begin{lemma}
In any homomorphism $h$ from $\Cylm$, with bottom cycle $\DCm$, to a reflexive tournament, if $|h(\DCm)|=1$, then $|h(\Cylm)|=1$.
\label{lem:claim1}
\end{lemma}

We will use another property, denoted \textbf{$(\dagger)$}, of $\Cylm$, which is that the retractions from $\Cylm$ to its bottom copy of $\DCm$, once propagated through the intermediate copies, induce on the top copy precisely the set of automorphisms of $\DCm$. That is, the top copy of $\DCm$ is mapped isomorphically to the bottom copy, and all such isomorphisms may be realised.
The reason is that in such a retraction, the $(j+1)$th copy may either map under the identity to the $j$th copy, or rotate one edge of the cycle clockwise, 
and $\Cylm$ consists of sufficiently many (namely $m$) copies of $\DCm$.

Now let $\H$ be a reflexive tournament that contains a subtournament~$\H_0$ on $m$ vertices that is endo-trivial.
By Lemma~\ref{l-Ham}, we find that $\H_0$ contains at least one directed Hamilton cycle $\HC_0$. 
Define $\mathrm{Spill}_m(\H[\H_0,\HC_0])$ as follows. Begin with $\H$ and add a copy of the gadget $\Cylm$, where the bottom copy of $\DCm$ is identified with $\HC_0$, to build a 
digraph~$\F(\H_0,\HC_0)$. Now ask, for some $y \in V(\H)$ whether there is a retraction $r$ of $\F(\H_0,\HC_0)$ to $\H$ so that some vertex $x$ in the top copy of $\DCm$ in $\Cylm$ is such that $r(x)=y$. Such vertices $y$ comprise the set $\mathrm{Spill}_m(\H[\H_0,\HC_0])$.

\medskip
\noindent
{\bf Remark~1.}
If $x$ belongs to some copy of $\DCm$ that is not the top copy, we can find a vertex~$x'$ in the top copy of $\DCm$ and a retraction $r'$ from $\F(\H_0,\HC_0)$ to $\H$ with $r'(x')=r(x)=y$, namely by letting $r'$ map the vertices of higher copies of~$\DCm$ to the image of their corresponding vertex in the copy that contains~$x$. In particular this implies that $\mathrm{Spill}_m(\H[\H_0,\HC_0])$ contains $V(\H_0)$.

\medskip
\noindent
%The set $\mathrm{Spill}_m(\H[\H_0,\HC_0])$ is potentially dependent on which Hamilton cycle in $\H_0$ is chosen.
We now observe that $\mathrm{Spill}_m(\H[\H_0,\HC_0])=V(\H)$ if $\H$ retracts to~$\H_0$.

%\begin{figure}
%\centering
%\input{lemma5}
%\caption{A diagram of the two cases of Lemma~\ref{lem:spillage}.}
%\label{fig:spillage}
%\end{figure}

\begin{lemma}
If $\H$ is a reflexive tournament that retracts to a 
subtournament~$\H_0$ with Hamilton cycle~$\HC_0$, then $\mathrm{Spill}_m(\H[\H_0,\HC_0])=V(\H)$.
\label{lem:spillage}
\end{lemma}
\begin{proof}
let $y \in V(\H) \setminus V(\H_0)$. We need to prove that there exists a retraction~$r$ from $\F(\H_0,\HC_0)$ to $\H$ with $r(x)=y$ for some vertex~$x$ in the top copy of $\DCm$ in $\Cylm$.
Let $h$ be a retraction of $\H$ to $\H_0$. Suppose $h(y)=i$. We observe that both $(i-1,y)$ and $(y,i+1)$ are edges of $E(\H)$. However, we might have either $(y,i)$ or $(i,y)$ and distinguish between these two cases.

First suppose that $(y,i) \in E(\H)$. Then we retract the gadget $\Cylm$ associated with $y$ in the following fashion. Using property $(\dagger)$ we turn successive copies of $\DCm$ in such a way as necessary to ensure that the vertex directly below $y$ is at position $i$; for a diagrammatic description of what means ``below'', see Figure~\ref{fig:Photo2}. Ih the top copy of $\DCm$ we map the $i$th vertex to $y$ and the $j$th vertex ($j\neq i$) to $j$ in $\H_0$.

Now suppose that $(i,y) \in E(\H)$.) Then we retract the gadget $\Cylm$ associated with $y$ in the following fashion. Using property $(\dagger)$ we turn successive copies of $\DCm$ in such a way as necessary to ensure that the vertex directly below $y$ is at position $i-1$; wherein with the last copy of $\DCm$ we map the $i$th vertex to $y$ and the $j$th vertex ($j\neq i$) to $j$ in $\H_0$.

Note that, in both cases, all of the vertices of $\Cylm$, except one, are mapped to $\H_0$.
\end{proof}

\subsection{Two Base Cases}
Recall that if $\H$ is an endo-trivial tournament, then {\sc Surjective $\H$-Colouring} is \NP-complete due to Corollary~\ref{c-main}. However $\H$ may not be endo-trivial. We will now show how to deal with the case
where $\H$ is not endo-trivial but retracts to an endo-trivial subtournament. For doing this we use the above gadget, but we need to distinguish between two different cases. 

\begin{figure}
\begin{center}
\resizebox{!}{3cm}{\includegraphics{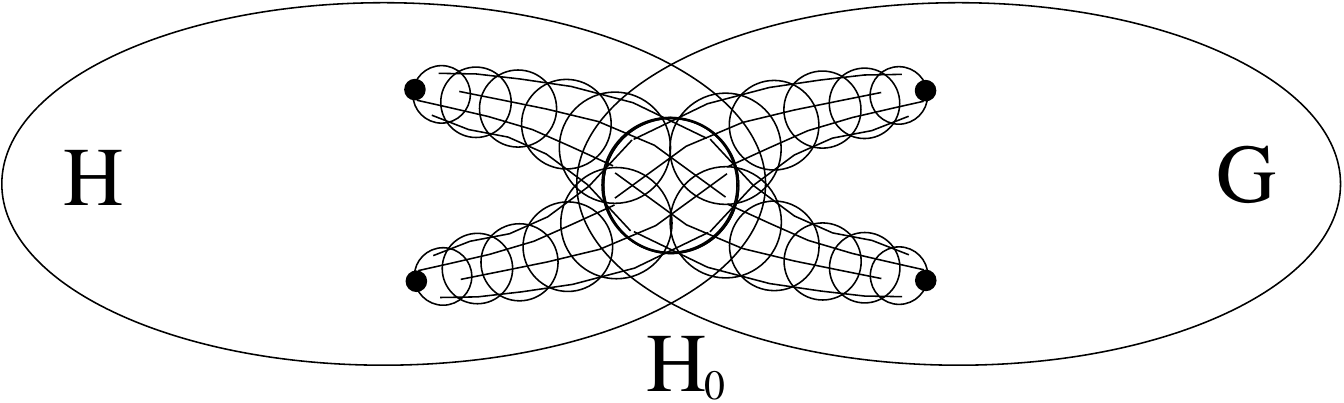}}
\end{center}
\caption{A stylised depiction of the construction in Base Case I. The central circle is the Hamilton cycle and the eccentric circles emanating thereout are the gadgets $\Cylm$.}\label{f-new}
\end{figure}

\begin{lemma}[Base Case I.]
Let $\H$ be a reflexive tournament that retracts to an endo-trivial subtournament $\H_0$ with Hamilton cycle $\HC_0$. 
Assume that $\H$ retracts to $\H'_0$ for every isomorphic copy $\H'_0=i(\H_0)$ of $\H_0$ in $\H$ with $\mathrm{Spill}_m(\H[\H'_0,i(\HC_0)])=V(\H)$.
Then $\H_0$-{\sc Retraction} can be polynomially reduced to {\sc Surjective $\H$-Colouring}.
\label{prop:main1}
\end{lemma}

\begin{proof}
Let $\G$ be an instance of $\H_0$-{\sc Retraction}.
We build an instance~$\G''$ of {\sc Surjective $\H$-Colouring}
in the following fashion. First, take a copy of $\H$ together with $\G$ and build $\G'$ by identifying these on the copy of $\H_0$ that they both possess as a subgraph.
Let $m$ be the size of $\H_0$ and consider its Hamilton cycle $\HC_0$. We build $\G''$ from $\G'$ by augmenting a new copy of $\Cylm$ for every vertex
$v \in V(\G') \setminus V(\H_0)$. 
Vertex~$v$ is to be identified with any vertex in the top copy of $\DCm$ in $\Cylm$ and the bottom copy of $\DCm$ is to be identified with $\HC_0$ in $\H_0$ according to the identity function.
See Figure~\ref{f-new} for an example.
We claim that $\G$ retracts to $\H_0$ if and only if there exists a surjective homomorphism from $\G''$  to $\H$.

First suppose that $\G$ retracts to $\H_0$.
Let $h$ be a retraction from $\G$ to $\H_0$. We extend $h$ as follows. First we map the copy of $\H$ in $\G''$ to itself in $\H$ by the identity. This will ensure surjectivity. We then map the various copies of $\Cylm$ in $\G''$. This is always possible: because $\H$ retracts to $\H_0$, we have $\mathrm{Spill}_m(\H[\H_0,\HC_0])=V(\H)$ due to Lemma~\ref{lem:spillage}. Hence, if $h(x)=y$ for two vertices $x\in V(\G')\setminus V(\H_0)$ and $y\in V(\H)$, we can always find a retraction of the graph $\F(\H_0,\HC_0)$ to $\H$ that maps $x$ to $y$, and we mimic this retraction on the corresponding subgraph in $\G''$. The crucial observation is that this can be done independently for each vertex in $V(\G')\setminus V(\H_0)$, as two vertices of different copies of $\Cylm$ are only adjacent if they both belong to $\G'$.
This leads to a surjective homomorphism from $\G''$ to $\H$.

Now suppose that there exists a surjective homomorphism~$h$ from $\G''$  to $\H$.
If $|h(\H_0)|=1$, then by Lemma~\ref{lem:claim1}, $|h(\Cylm)|=1$ for all copies of $\Cylm$ in $\G''$. This means that $|h(\G'')|=1$ and $h$ is not surjective, a contradiction. Now, $1<|h(\H_0)|<m$ is not possible either due to Lemma~\ref{lem:under-m}. Thus, $|h(\H_0)|=m$ and indeed $h$ maps $\H_0$ to a copy of itself in $\H$ which we will call $\H'_0=i(\H_0)$ for some isomorphism~$i$. 

We claim that $\mathrm{Spill}_m(\H[\H'_0,i(\HC_0)]) = V(\H)$. In order to see this, consider a vertex $y\in V(\H)$. 
As $h$ is surjective, there exists a vertex~$x\in V(\G'')$ with $h(x)=y$. By construction, $x$ belongs to some copy of $\DCm$, and thus also belongs to some copy of $\DCm$ in $\F(\H_0,\HC_0)$.
We can extend $i^{-1}$ to an isomorphism from the copy of $\Cylm$  (which has $i(\HC_0)$ as its bottom cycle)  in the graph $\F(\H_0',i(\HC_0))$ to the copy of $\Cylm$ (which has $\HC_0$ as its bottom cycle) in the graph $\F(\H_0,\HC_0)$.
We define a mapping~$r^*$ from $\F(\H_0',i(\HC_0))$ to $\H$ by $r^*(u)=h\circ i^{-1}(u)$ if $u$ is on the copy of $\Cylm$ in $\F(\H_0',i(\HC_0))$ and $r^*(u)=u$ otherwise. We observe that $r^*(u)=u$ if $u\in V(\H_0')$ as $h$ coincides with $i$ on $\H_0$. As $\H_0$ separates the other vertices of 
the copy of $\Cylm$ from $V(\H)\setminus V(\H_0)$, in the sense that removing $\H_0$ would disconnect them, this means that $r^*$ is a retraction from $\F(\H_0',i(\HC_0))$ to $\H$. We find that $r^*$ maps $i(x)$ to $h\circ i^{-1}(i(x))=h(x)=y$. Moreover, as $x$ is in some copy of $\DCm$ in $\F(\H_0,\HC_0)$, we have that $i(x)$ is in some copy of $\DCm$ in $\F(\H'_0,i(\HC_0))$. We may assume without loss of generality that $i(x)$ belongs to the top copy (cf. Remark~1). We conclude that $y$ always belongs to  $\mathrm{Spill}_m(\H[\H'_0,i(\HC_0)])$ (cf. Remark~1).

As $\mathrm{Spill}_m(\H[\H'_0,i(\HC_0)]) = V(\H)$, we find, by assumption of the lemma, that there exists a retraction $r$ from $\H$ to $\H'_0$.  
Now $i^{-1} \circ r \circ h$ ia the desired retraction of $\G$ to~$\H_0$.  
\end{proof}

\begin{figure}
\centering
\input{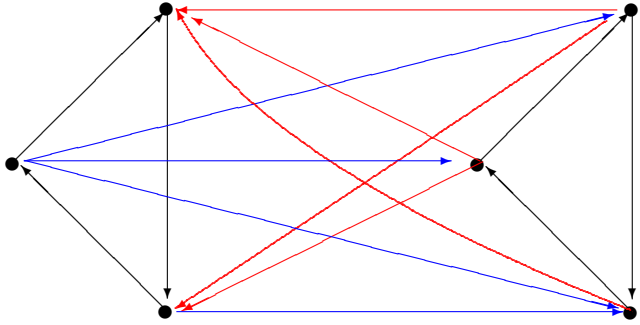}
\caption{An interesting tournament $\H$ on six vertices (self-loops are not drawn). This tournament does not retract to the $\DC^*_3$ on the left-hand side, yet $\mathrm{Spill}_3(\H[\DC^*_3,\DC_3]) = V(\H)$.}
\label{fig:weird-example}
\end{figure}

%%%% The folowing lemma seems to not be used 
%\begin{lemma}
%Let $\H$ be a reflexive tournament containing a subtournament $\H_0$ so that $\H$ has an endomorphism $h$ mapping $\H_0$ to an isomorphic copy $\H'_0=h(\H_0)$ of $\H_0$. Then there is an isomorphic copy of $h(\H)$ in $\H$ containing $\H_0$.
%\label{lem:6}
%\end{lemma}
%\begin{proof}
%Choose $h^{-1}$ in the following fashion. $h^{-1}$ of $h(\H_0)$ is the natural isomorphism of $h(\H_0)$ to $\H_0$ (that inverts the isomorphism given by $h$ from $\H_0$ to $\H'_0$). Otherwise choose $h^{-1}$ arbitrarily, such that $h^{-1}(y)=x$ only if $h(x)=y$. Since $\H$ is a (reflexive) tournament, $h^{-1}$ is an isomorphism.
%\end{proof}

We now need to deal with the situation in which we have an isomorphic copy $\H'_0=i(\H_0)$ of $\H_0$ in $\H$ with $\mathrm{Spill}_m(\H[\H'_0,i(\HC_0)])=V(\H)$, such that $\H$ does not retract to $\H'_0$ (see Figure~\ref{fig:weird-example} for an example). We cannot deal with this case in a direct matter and first show another base case. For this we need the following lemma and an extension of endo-triviality that we discuss afterwards.

\begin{lemma}
Let $\H$ be a reflexive tournament, containing a subtournament $\H_0$ so that any endomorphism of $\H$ that fixes $\H_0$ is an automorphism. Then any endomorphism of $\H$ that maps $\H_0$ to an isomorphic copy $\H'_0=i(\H_0)$ of itself is an automorphism of $\H$. 
\label{lem:7}
\end{lemma}
\begin{proof}
For contradiction, suppose there is an endomorphism $h$ that maps $\H_0$ to an isomorphic copy $\H'_0=i(\H_0)$ of itself that is not an automorphism of $\H$. In particular, $|h(\H)|<|V(\H)|$. Choose $h^{-1}$ in the following fashion. We let $h^{-1}$ of $h(\H_0)$ be the natural isomorphism of $h(\H_0)$ to $\H_0$ (that inverts the isomorphism given by $h$ from $\H_0$ to $\H'_0$). Otherwise we choose $h^{-1}$ arbitrarily, such that $h^{-1}(y)=x$ only if $h(x)=y$. Since $\H$ is a reflexive tournament, containing precisely one edge between distinct vertices, $h^{-1}$ is an isomorphism. Moreover, $h^{-1} \circ h$ is an endomorphism of $\H$ that fixes $\H_0$ and that is not an automorphism, a contradiction.
\end{proof}

Let $\H_0$ be an induced subgraph of a digraph $\H$. We say that the pair $(\H,\H_0)$ is \emph{endo-trivial} if all endomorphisms of $\H$ that fix $\H_0$ are automorphisms.

\begin{lemma}[Base Case II]
Let $\H$ be a reflexive tournament with a subtournament~$\H_0$ with Hamilton cycle $\HC_0$ so that $(\H,\H_0)$ and $\H_0$ are endo-trivial and $\mathrm{Spill}_m(\H[\H_0,\HC_0]) = V(\H)$. 
Then $\H$-{\sc Retraction} can be polynomially reduced to {\sc Surjective $\H$-Colouring}.
\label{prop:main2}
\end{lemma}
\begin{proof}
Let $\G$ be an instance of {\sc $\H$-Retraction}. 
We build an instance $\G'$ of {\sc Surjective $\H$-Colouring}
in the following fashion. First we build $\G'$ from $\G$ by augmenting a new copy of $\Cylm$ for every vertex $v \in V(\G) \setminus V(\H_0)$. Vertex~$v$ is to be identified with any vertex in the top copy of $\DCm$ in $\Cylm$ and the bottom copy of $\DCm$ is to be identified with $\HC_0$ in $\H_0$ according to the identity function. We claim that $\G$ retracts to $\H$ if and only if there exists a surjective homomorphism from $\G'$ to $\H$.

First suppose $\G$ retracts to $\H$. Let $r$ be a retraction from $\G$ to $\H$. Then any extension of~$r$ from $\G$ to $\G'$ is surjective. As $\mathrm{Spill}_m(\H[\H_0,\HC_0]) = V(\H)$ and two vertices in different copies of $\Cylm$ are only adjacent if both of them are in $\G$, we can in fact extend $r$ to a surjective homomorphism from $\G'$ to $\H$.

Now suppose there exists a surjective homomorphism~$h$ from $\G'$ to $\H$.
If $|h(\H_0)|=1$, then Lemma~\ref{lem:claim1} tells us that  $|h(\Cylm)|=1$ for all copies of $\Cylm$ in $\G'$, and then we derive $|h(\G')|=1$, contradicting the surjectivity of~$h$.
Moreover, $1<|h(\H_0)|<m$ is not possible either due to Lemma~\ref{lem:under-m}. Thus, $|h(\H_0)|=m$ and $h$ maps $\H_0$ to a copy of itself. 
As $(\H,\H_0)$ is endo-trivial, Lemma~\ref{lem:7} tells us that  the restriction of $h$ to $\H$ 
is an automorphism of $\H$, which we call $\alpha$. The required retraction from $\G$ to $\H$ is now given by $\alpha^{-1} \circ h$.
\end{proof}

\subsection{Generalising the Base Cases}\label{s-general}

We now generalise the two base cases to more general cases via some recursive procedure. Afterwards we will show how to combine these two cases to complete our proof.
We will first need a slightly generalised version of Lemma~\ref{lem:7}, which nonetheless has virtually the same proof.

\begin{lemma}
Let $\H_2 \supset \H_1 \supset H_0$ be a sequence of strongly connected reflexive tournaments, each one a subtournament of the one before. Suppose that any endomorphism of $\H_1$ that fixes $\H_0$ is an automorphism. Then any endomorphism $h$ of $\H_2$ that maps $\H_0$ to an isomorphic copy $\H'_0=i(\H_0)$ of itself also gives an isomorphic copy of $\H_1$ in $h(\H_1)$. 
\label{lem:8}
\end{lemma}
\begin{proof}
For contradiction, suppose there is an endomorphism $h$ of $\H_2$ that maps $\H_0$ to an isomorphic copy $\H'_0=i(\H_0)$ of itself that does not yield an isomorphic copy of $\H_1$. In particular, $|h(\H_1)|<|V(\H_1)|$. We proceed as in the proof of the Lemma~\ref{lem:7}. Choose $h^{-1}$ in the following fashion. We let $h^{-1}$ of $h(\H_0)$ be the natural isomorphism of $h(\H_0)$ to $\H_0$ (that inverts the isomorphism given by $h$ from $\H_0$ to $\H'_0$). Otherwise we choose $h^{-1}$ arbitrarily, such that $h^{-1}(y)=x$ only if $h(x)=y$. Since 
$\H_2$ is a reflexive tournament, $h^{-1}$ is an isomorphism. And $h^{-1} \circ h$ is an endomorphism of $\H_2$ that fixes $\H_0$ that does not yield an isomorphic copy of $\H_1$ in $h(\H_1)$, a contradiction.
\end{proof}

The following two lemmas generalize Lemmas~\ref{prop:main1} and \ref{prop:main2}.

\begin{lemma}[General Case I]
Let $\H_{0},\H_{1}, \ldots, \H_k, \H_{k+1}$
be reflexive tournaments, the first $k$ of which have Hamilton cycles $\HC_{0},\HC_{1}, \ldots, \HC_{k}$, respectively,
so that $\H_0 \subseteq H_1 \subseteq \cdots \subseteq \H_k \subseteq \H_{k+1}.$
Assume that $\H_0$, $(\H_1,\H_0)$, \ldots, $(\H_{k},\H_{k-1})$ are endo-trivial and that
\[
\begin{array}{lcl}
\mathrm{Spill}_{a_0}(\H_1[\H_0,\HC_{0}]) &= &V(\H_1) \\
\mathrm{Spill}_{a_1}(\H_2[\H_1,\HC_{1}]) &= &V(\H_2) \\
\hspace*{3mm} \vdots &\vdots &\hspace*{3mm}\vdots \\
\mathrm{Spill}_{a_{k-1}}(\H_{k}[\H_{k-1},\HC_{k-1}]) &= &V(\H_k).\\
\end{array}
\]
Assume that $\H_{k+1}$ retracts to $\H_k$ and also to every isomorphic copy $\H'_{k}=i(\H_k)$ of $\H_k$ in $\H_{k+1}$ with 
$\mathrm{Spill}_{a_k}(\H_{k+1}[\H'_k,i(\HC_{k})]) = V(\H_{k+1})$.
Then $\H_k$-{\sc Retraction} can be polynomially reduced to {\sc Surjective $\H_{k+1}$-Colouring}.
\label{prop:main-general1}
\end{lemma}
\begin{proof}
Let $\G$ be an instance of {\sc $\H_k$-Retraction}.
We will build an instance $\G''$ of {\sc Surjective $\H_{k+1}$-Colouring} in the following fashion. First, take a copy of $\H_{k+1}$ together with $\G$ and build $\G'$ by identifying these on the copy of $\H_{k}$ that they both possess as a subgraph.
We now build $\G''$ as follows. 
First we augment $\G'$ with a new copy of $\Cyl_{a_k}$ for every vertex $v\in V(\G')\setminus V(\H_k)$. Vertex~$v$ is to be identified with any vertex in the top copy of 
$\DC^*_{a_k}$ in $\Cyl_{a_k}$, and the bottom copy of $\DC^*_{a_k}$ is to be identified with $\HC_k$ according to the identity function.
Then, for each $i \in [k+1]$, and $v \in V(\H_{i}) \setminus V(\H_{i-1})$, add a copy of $\Cyl_{a_{i-1}}$, where $v$ is identified with any vertex in the top copy of $\DC^*_{a_{i-1}}$ in $\Cyl_{a_{i-1}}$ and the bottom copy of $\DC^*_{i-1}$ is to be identified with $\H_{i-1}$ according to the identity map of $\DC^*_{a_{i-1}}$ to $\HC_{i-1}$.
We claim that $\G$ retracts to $\H_k$ if and only if there exists a surjective homomorphism from $\G''$ to $\H_{k+1}$.

First suppose that $\G$ retracts to $\H_k$.
Let $h$ be a retraction from $\G$ to $\H_{k}$. Extend $h$ mapping $\H_{k+1}$ according to the identity to ensure the final mapping is surjective. Finally, we map the various copies of $\Cyl_{a_{i-1}}$ in $\G''$ in any suitable fashion, which will always exist due to our assumptions and the fact that $\mathrm{Spill}_{a_{k}}(\H_{k+1}[\H_{k},\HC_{k}]) = V(\H_{k+1})$, which follows from our assumption that $\H_{k+1}$ retracts to $\H_k$ and Lemma~\ref{lem:spillage}.

Now suppose that if there exists a surjective homomorphism~$h$ from $\G''$ to $\H_{k+1}$.
Suppose that $|h(\H_0)|=1$. Then 
$|h(\Cyl_{a_0})|=1$
 by Lemma~\ref{lem:claim1}. Now we follow the chain of spills to deduce that $|h(\H_{k+1})|=1$, which is not possible.
Now, $1<|h(\H_0)|<a_0$ is not possible either due to Lemma~\ref{lem:under-m}. 
Thus, $|h(\H_0)|=|V(\H_0)|$ and indeed $h$ maps $\H_0$ to a copy of itself in $\H_{k+1}$ which we will call $\H'_0=i(\H_0)$. We now apply Lemma~\ref{lem:8} as well as our assumed endo-trivialities to derive that $h$ in fact maps $\H_k$ by the isomorphism $i$ to a copy of itself in $\H_{k+1}$ which we will call $\H'_{k}$. Since $h$ is surjective, we can deduce that $\mathrm{Spill}_{a_k}(\H_{k+1}[\H'_k,i(\HC_{k})]) = V(\H_{k+1})$
in the same way as in the proof of Lemma~\ref{prop:main1}.
 and so there exists a retraction $r$ from $\H_{k+1}$ to $\H'_k$.  Now $i^{-1} \circ r \circ h$ gives the desired retraction of $\G''$ to $\H_k$.   
\end{proof}

\begin{lemma}[General Case II]
\label{prop:main-general2}
Let  $\H_{0},\H_{1}, \ldots, \H_k, \H_{k+1}$ be reflexive tournaments, the first $k+1$ of which have Hamilton cycles $\HC_{0},\HC_{1}, \ldots, \HC_{k}$, respectively,
so that $\H_0 \subseteq H_1 \subseteq \cdots \subseteq \H_k \subseteq \H_{k+1}$.
Suppose that $\H_0$, $(\H_1,\H_0)$, \ldots, $(\H_{k},\H_{k-1}), (\H_{k+1},\H_{k})$ are endo-trivial and that
\[
\begin{array}{lcl}
\mathrm{Spill}_{a_0}(\H_1[\H_0,\HC_{0}]) &= &V(\H_1) \\
\mathrm{Spill}_{a_1}(\H_2[\H_1,\HC_{1}]) &= &V(\H_2) \\
\hspace*{3mm} \vdots &\vdots &\hspace*{3mm}\vdots \\
\mathrm{Spill}_{a_{k-1}}(\H_{k}[\H_{k-1},\HC_{k-1}]) &= &V(\H_k)\\
\mathrm{Spill}_{a_{k}}(\H_{k+1}[\H_{k},\HC_{k}]) &= &V(\H_{k+1})
\end{array}
\]
Then {\sc $\H_{k+1}$-Retraction} can be polynomially reduced to {\sc Surjective $\H_{k+1}$-Colouring}.
\end{lemma}
\begin{proof}
Let $\G$ be an instance of {\sc $\H_{k+1}$-Retraction}.
We build an instance $\G'$ of {\sc Surjective $\H_{k+1}$-Colouring}
in the following fashion. Build $\G'$ by, for each $i \in [k+1]$, and $v \in V(\H_{i}) \setminus V(\H_{i-1})$, adding a copy of $\Cyl_{a_{i-1}}$, where $v$ is identified with any vertex in the top copy of $\DC^*_{a_{i-1}}$ in $\Cyl_{a_{i-1}}$ and the bottom copy of $\DC^*_{i-1}$ is to be identified with $\H_{i-1}$ according to the identity map of $\DC^*_{a_{i-1}}$ to $\HC_{i-1}$.
We claim that $\G$ retracts to $\H_{k+1}$ if and only if there exists a surjective homomorphism from $\G'$ to $\H_{k+1}$.

First suppose that $\G$ retracts to $\H_{k+1}$.
Let $h$ be a retraction from $\G$ to $\H_{k+1}$. Then we can extend $h$ by mapping $\G'$ in some suitable fashion, which is possible due to the spill assumptions.

Now suppose that there exists a surjective homomorphism~$h$ from $\G'$ to $\H_{k+1}$. 
Suppose that $|h(\H_0)|=1$. Then $|h(\Cyl_{a_0})|=1$ by Lemma~\ref{lem:claim1}. Now we follow the chain of spills to deduce that $|h(\H_{k+1})|=1$, a contradiction. 
Now, $1<|h(\H_0)|<a_0$ is not possible either due to Lemma~\ref{lem:under-m}. 
Thus, $|h(\H_0)|=|V(\H_0)|$ and indeed $h$ maps $\H_0$ to a copy of itself in $\H_{k+1}$ which we will call $\H'_0=i(\H_0)$. We now apply Lemma~\ref{lem:8} as well as our assumed endo-trivialities to derive that $h$ in fact maps $\H_k$ by the isomorphism $i$ to a copy of itself in $\H_{k+1}$, which we will call $\H'_{k}$. Now we can deduce, via Lemma~\ref{lem:7}, that $h(\H_{k+1})$ is an automorphism of $\H_{k+1}$, which we call $\alpha$. The required retraction from $\G'$ to 
$\H_{k+1}$ 
is now given by $\alpha^{-1} \circ h$.
\end{proof}

\subsection{Final Steps for Hardness for Non-Transitive Reflexive Tournaments}

We first prove, by using the lemmas from Section~\ref{s-general}, that {\sc Surjective $\H$-Colouring} is \NP-complete if $\H$ is a non-transitive reflexive tournament that is strongly connected.
For our discourse it is not necessary to know precisely what is a Taylor operation, but we will use the following result.

\begin{theorem}[\cite{JBK,Larose2005}]
Let $\H$ be a finite structure so that the idempotent polymorphisms of $\H$ omit all Taylor operations. Then 
{\sc $\H$-Retraction} is \NP-complete.
\label{thm:ret-hardness}
\end{theorem}

\begin{corollary}
Let $\H$ be a strongly connected reflexive tournament. 
Then {\sc Surjective $\H$-Colouring} is \NP-complete.
\label{cor:strongly-connected1}
\end{corollary}
\begin{proof}
As $\H$ is is a strongly connected reflexive tournament, which has more than one vertex by our definition, $\H$ is not transitive.
Note that {\sc $\H$-Retraction}
 is \NP-complete, since non-transitive reflexive tournaments omit Taylor polymorphisms~\cite{larose2006taylor}, following Theorem~\ref{thm:ret-hardness}. Thus, if $\H$ is endo-trivial, the result follows from Lemma~\ref{prop:main1} (note that we could also have used Corollary~\ref{c-main}).

Suppose $\H$ is not endo-trivial. Then, by Lemma~\ref{lem:endo-retraction-trivial}, $\H$ is not retract-trivial either. This means $\H$ has a non-trivial retraction to some subtournament~$\H_0$. We may assume
that $\H_0$ is endo-trivial, as otherwise we will repeat the argument until we find a retraction from $\H$ to an endo-trivial subtournament.
 % , choose some endo-trivial $\H_0$ so that $\H$ retracts to $\H_0$. One might argue to choose $\H_0$ in some way such that it is maximal according to some (partial) order, but this will not be important. 
 
If $\H$ retracts to all isomorphic copies $\H'_0=i(\H_0)$ of $\H_0$ within it, except possibly those for which $\mathrm{Spill}_m(\H[\H'_0,i(\HC_0)]) \neq V(\H)$, then the result follows from  Lemma~\ref{prop:main1}. So there is a copy $\H'_0=i(\H'_0)$ to which $\H$ does not retract for which $\mathrm{Spill}_m(\H[\H'_0,i(\HC_0)]) = V(\H)$.
If $(\H,\H'_0)$ is endo-trivial, the result follows from Lemma~\ref{prop:main2}. Thus we assume $(\H,\H'_0)$ is not endo-trivial and we deduce the existence of $\H'_0 \subset \H_1 \subset \H$ ($\H_1$ is strictly between $\H$ and $\H'_0$) so that $(\H_1,\H'_0)$ is endotrivial and $\H$ retracts to $\H_1$. Now we are ready to break out. Either $\H$ retracts to all isomorphic copies of $\H'_1=i(\H_1)$ in $\H$, except possibly for those so that $\mathrm{Spill}_m(\H[\H'_1,i(\HC_1)]) \neq V(\H)$, and we apply Lemma~\ref{prop:main-general1}, or there exists a copy $\H'_1$, with $\mathrm{Spill}_m(\H[\H'_1,i(\HC_1)]) = V(\H)$, to which it does not retract. Then $\H'_1$ contains $\H''_0$ a copy of $\H''_0$ so that $(\H'_1,\H''_0)$ and $\H''_0$ are endo-trivial. We now continue iterating this method, which will terminate because our structures are getting strictly larger.
\end{proof}

In order to deal with reflexive tournaments that are not strongly connected we
need the following strengthened version of Corollary~\ref{cor:strongly-connected1}.
\begin{corollary}
Let $\H$ be a strongly connected reflexive tournament. 
Then  {\sc Surjective $\H$-Colouring}
is \NP-complete
even for strongly connected digraphs.
\label{cor:strongly-connected2}
\end{corollary}
\begin{proof}
We need to argue that the instances of {\sc Surjective $\H$-Colouring}
that we have constructed 
before
can be assumed to be strongly connected. Noting that $\H$ and the gadgets $\Cylm$ are strongly connected, this is clear once we can assume the inputs to our Retraction problems are strongly connected. For 
{\sc $\H'$-Retraction},
where $\H'$ is a strongly connected reflexive tournament, we can surely assume our inputs are strongly connected. If they were not, then we add individual directed paths of length $|V(\H')|$ between the relevant vertices. This will not affect the truth of an instance and the result follows.
\end{proof}
When $\G$ is a reflexive tournament, we may break it up into strongly connected components 
$\G(1),\ldots,\G(k)$ so that, for all $i < j \in [k]$, for all $x \in \G(i)$ and for all $y \in \G(j)$, $(x,y) \in E(\G)$. This is the standard decomposition, inducing a standard order on the connected components, that we will use.

We now prove our main hardness result.

\begin{theorem}
Let $\H$ be a non-transitive reflexive tournament. Then {\sc Surjective $\H$-Colouring}
is \NP-complete.
\label{thm:hauptsatz}
\end{theorem}
\begin{proof}
For strongly connected tournaments, the result follows from Corollary~\ref{cor:strongly-connected1}. Let $\H$ instead have $k>1$ strongly connected components
$\H(1),\ldots,\H(k)$. 
Since $\H$ is not transitive, one of these strongly connected components, $\H(i)$, must be of size greater than $1$, whereupon we know from Corollary~\ref{cor:strongly-connected2} that
{\sc Surjective $\H(i)$-Colouring}
 is \NP-complete, even when restricted to strongly connected inputs. 
 
 Let us reduce
{\sc Surjective $\H(i)$-Colouring} to {\sc Surjective $\H$-Colouring} 
by taking a strongly connected input $\G$ for the former and building $\G'$ by adding a copy of $\H$ restricted to $V(\H(1)),\ldots,V(\H(i-1))$, where every vertex here has an edge to every vertex of $\G$, and adding a copy of  $\H$ restricted to $V(\H(i+1)),\ldots,V(\H(k))$, where every vertex there has an edge from every vertex of $\G$. 
Note that $\G'$ has $k$ strongly connected components $\G'(1),\ldots, \G'(k)$, where $\G'(h)$ is isomorphic to $\H'(h)$ for $h=1,\ldots,k$, $h\neq i$.
We claim 
that there exists a surjective homomorphism from $\G$ to $\H(i)$ if and only if there exists a surjective homomorphism from $\G'$ to 
$\H$.

(Forwards.) Map the additional vertices in $\G'$ in the obvious fashion (by the ``identity'') to extend a surjective homomorphism
from $\G$ to $\H_i$ so that it is surjective from $\G'$ to
$\H$.

(Backwards.) 
In any surjective homomorphism $s$ from $\G'$ to $\H$, we have that $s(V(\G'(i))) \subseteq V(\H(i))$ for $i=1,\ldots k$. 
\end{proof}

\begin{corollary}
Let $\H$ be a reflexive tournament. 
If $\H$ is transitive, then {\sc Surjective $\H$-Colouring}
 is in \NL; otherwise it is \NP-complete.
\label{cor:hauptsatz}
\end{corollary}
\begin{proof}
For the transitive case we can say that {\sc \H-Retraction} is in \NL\ from \cite{DalmauK08}, since $\H$ enjoys the ternary median operation as a polymorphism (this has been observed, inter alia, in \cite{larose2006taylor}). It follows of course that {\sc Surjective $\H$-Colouring} is in \NL\ also. The non-transitive case follows from Theorem~\ref{thm:hauptsatz}.
\end{proof}

\section{Digraphs with a most three vertices}

In this section we prove the following result.
\begin{theorem}
Let $\H$ be a partially reflexive digraph of size at most $3$. Then {\sc Surjective $\H$-Colouring} is polynomially equivalent to {\sc $\H$-Retraction}. In particular, it is always in \Ptime\ or is \NP-complete.
\label{thm:3digraph}
\end{theorem}

We are not aware of a published classification for {\sc $\H$-Retraction}, when $\H$ is a partially reflexive digraph of size at most $3$, though we know of one for {\sc List $\H$-Colouring} from \cite{FederHellTucker}. Our starting point is therefore 
Theorem~3.1 from \cite{FederHellTucker}, and in particular the sporadic digraphs drawn in Figure~\ref{fig:3digraphs} that are precisely those for which {\sc List $\H$-Colouring} is not in \Ptime. Bearing in mind that membership in \Ptime\ for {\sc List $\H$-Colouring} gives this a fortiori for {\sc $\H$-Retraction} and {\sc Surjective $\H$-Colouring}, these sporadic digraphs are the only ones we need to consider. Note that the principal objects of study in \cite{FederHellTucker} are \emph{trigraphs} and the reference to \emph{complement} in that paper's Theorem~3.1 is to trigraph complement, which is different from the various notions of (di)graph complement.

\begin{figure}
%\begin{center}
\includegraphics[scale=0.6,trim={0 4.7cm 0 0},clip]{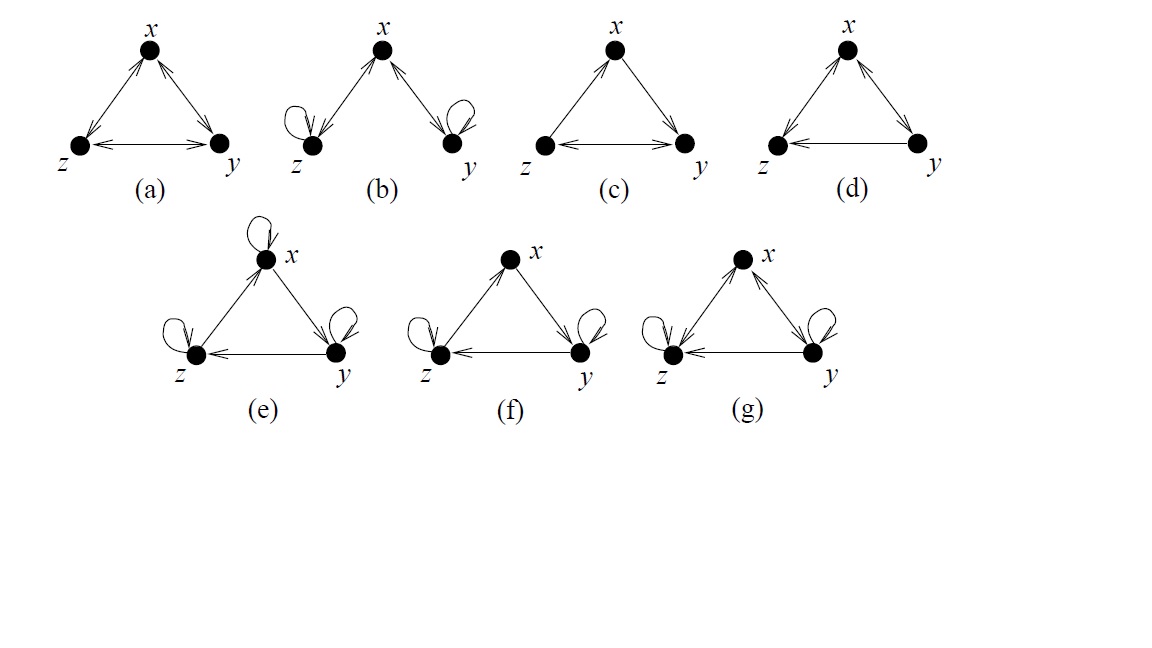}
%\end{center}
\caption{The \NP-hard cases from Figure~1 in \cite{FederHellTucker}}
\label{fig:3digraphs}
\end{figure}
We will need two lemmas to deal with two of the cases from Figure~\ref{fig:3digraphs}.
\begin{lemma}
Let $\H$ be the digraph from Figure~\ref{fig:3digraphs} (f). All polymorphisms of $\H$ are essentially unary.
\label{lem:case-f}
\end{lemma}
\begin{proof}
Recall the self-map digraph $\G^{\G}$, with vertices the self-maps of $\G$, and an edge $(f,g) \in E(\G^\G)$ between self-maps $f$ and $g$ if and only if for every edge $(x,y) \in E(\G)$, we have that $(f(x),g(y)) \in E(\G)$. %The true endomorphism digraph is the induced subgraph of this where the self-maps are additionally endomorphisms. 
In the self-map digraph, the endomorphisms are precisely the looped vertices.

Consider the homomorphism $\phi: \H^n \rightarrow \H^{\H}$ induced by an $n+1$-ary polymorphism of $\H$ where $n \geq 1$; suppose for a contradiction that the polymorphism depends on all its variables. Then clearly the image of $\phi$ contains at least two elements, and at least one of these elements is not a constant map.  

\noindent{\bf Claim 0.} {\em The only loops in $\H^{\H}$ are the constant maps 1 and 2 and the identity. }

\noindent {\em Proof of Claim 0}: let $f$ be an endomorphism of $\H$. Then $f$ maps loops to loops. Since $(1,2)\in E(\H)$, $f$ can either map both to 1, both to 2 , or fix both. Clearly if it fixes both $f$ is the identity. Otherwise, if $f(1)=f(2)=i$, then $(f(0),i) , (i,f(0)) \in E(\H)$ so $f(0)=i$. 

Without fear of confusion, we denote the constant maps by 0, 1 and 2. We denote the identity map by $id$.

\noindent{\bf Claim 1.} {\em (a) There are no edges in  $\H^{\H}$ between the identity and the constant maps; (b) if $(2,f), (f,1) \in E(\H^\H)$, then $f=0$; (c) if $c$ is a constant map such that $(f,c),(c,f) \in E(\H^\H)$ then $f = c$. }

\noindent {\em Proof of Claim 1}:  $(2,f) \in E(\H^\H)$ implies $(2,f(i)) \in E(\H)$ for all $i$ so $f(i) \in \{0,2\}$ for all $i$. Similarly $(f,1)\in E(\H^\H)$ implies $f(i) \in \{0,1\}$ for all $i$ so (b) follows immediately. To prove (a), observe that any in- or out-neighbour of a constant cannot be surjective since no constant has in- or out-neighbourhood of size 3. For (c), argue again as we just did: $(1,f) \in E(\H^\H)$ means $f(i) \in \{1,2\}$ for all $i$, and $(f,1)\in E(\H^\H)$ implies $f(i) \in \{0,1\}$ for all $i$ and hence $f=1$; the proof for $c=2$ is identical.

\noindent{\bf Claim 2.} {\em If $(f,id),(id,f) \in E(\H^\H)$ then $f = id$. }

\noindent {\em Proof of Claim 2}: Assume that $(f,id),(id,f) \in E(\H^\H)$. If $i$ is a loop then $(i,i)\in E(\H)$ implies $(f(i),i),(i,f(i))\in E(\H)$  thus $f(i) = i$. Now $(0,1)\in E(\H)$ implies $(f(0),1)\in E(\H)$, and $(2,0)\in E(\H)$ implies $(2,f(0))\in E(\H)$ so $f(0)=0$.

Obviously the loops of $\H^n$ are precisely the tuples whose coordinates belong to $\{1,2\}$. Call the subdigraph induced by these vertices $\W$, and notice it is weakly connected. Then the homomorphism $\phi$ must map $\W$  onto a weakly connected digraph consisting of loops only. By Claim 1 (a) either  (i) $\W$ is entirely mapped to $\{id\}$ or (ii) $\W$ is mapped by $\phi$ to the set $\{1,2\}$.

Choose any tuple $X$ of $\H^n$ containing a coordinate equal to 0. Then there exist $Y, Z \in W$ such that $(X,Z), (Z,Y), (Y,X) \in E(\H^n)$: indeed, let $Z$ be obtained from $X$ by replacing each  0 entry by 1, fixing all other coordinates, and let $Y$ be obtained from $X$ by replacing each 0 entry by 2, fixing all other coordinates. Thus in case (i) we get that $id = \phi(Y), (\phi(Y),\phi(X)),(\phi(X),\phi(Z)) \in E(\H^\H)$ and $\phi(Z)=id$ so by Claim 2 $\phi$ maps all of $\H^n$ to $id$. In case (ii), we get that $\phi(Y)$ and $\phi(Z)$ belong to $\{1,2\}$; since $(\phi(Z),\phi(Y))\in E(\H^\H)$ we must in fact have that $(\phi(Z),\phi(Y))$ belongs to $\{(1,1),(1,2),(2,2)\}$; thus we get one of the following cases, namely $(1,\phi(X)),(\phi(X),1) \in E(\H^\H)$ or $(1,\phi(X)),(\phi(X),2) \in E(\H^\H)$ or $(2,\phi(X)),(\phi(X),2) \in E(\H^\H)$. By Claim 1 (b) and (c)  $\phi(X)$ must be one of the constant maps 0, 1 or 2. 

Thus we conclude that $\phi$ either maps all of $\H^n$ to $\{id\}$, in which case our polymorphism cannot depend on its first $n$ variables; or otherwise $\phi$ maps all of $\H^n$ to constant maps, and then the polymorphism does not depend on its last variable, also a contradiction.  

%(The argument, recall, is that if \phi maps all of G^n to id then the polymorphism does not depend on its first n variables; and it phi maps G^n to cst maps only, then %the polymorphism does not depend on its last variable.)
\end{proof}
\noindent An operation $t:D^k \rightarrow D$ is a {\em weak near-unanimity} operation if $t$ satisfies
\[
\begin{array}{c} 
t(x,\ldots,x)=x, \mbox{ and} \\
t(y,x,\dots,x)=t(x,y,x,\dots,x)=\dots=t(x,\dots,x,y).
\end{array}
\]
\begin{lemma}
The digraph~$\H$ from Figure~\ref{fig:3digraphs} (g) admits a weak near-unanimity polymorphism.
\label{lem:case-g}
\end{lemma}
\begin{proof}
Let $\H$ be given over vertex-set $\{0,1,2\}$ with edge-set 
\[\{(0,0), (2,2), (0,1), (1,0), (0,2), (2,0), (2,1)\}.\]
We have found by computer the following ternary weak near-unanimity polymorphism.
\[
\begin{array}{ccccccccc}
000=0 &
001=0 &
002=0 &
010=0 &
011=0 &
012=0 &
020=0 &
021=0 &
022=0 
\\
100=0 &
101=0 &
102=0 &
110=0 &
111=1 &
112=0 &
120=0 &
121=0 &
122=0 
\\
200=0 &
201=0 &
202=0 &
210=0 &
211=0 &
212=0 &
220=0 &
221=0 &
222=2
\end{array}
\]
We have used to find this polymorphism the excellent program of Mikl\'os Mar\'oti.\footnote{See: \texttt{http://www.math.u-szeged.hu{\raise.17ex\hbox{$\scriptstyle\sim$}}maroti/applets/GraphPoly.html}}
\end{proof}
For the benefit of the interested reader, we note that a finite core possesses a weak near-unanimity polymorphism if and only if it possesses a Taylor polymorphism.
\begin{proof}[Proof of Theorem~\ref{thm:3digraph}]
As discussed, we only need to settle the complexity of {\sc $\H$-Retraction} and {\sc Surjective $\H$-Colouring} for the digraphs depicted in Figure~\ref{fig:3digraphs}, as for all other partially reflexive digraphs of size at most $3$, {\sc List $\H$-Colouring} is in P \cite{FederHellTucker}.

When $\H$ is as depicted in Figure~\ref{fig:3digraphs} (a), (c) or (d), then $\H$ is an irreflexive semicomplete digraph with more than one cycle and it is known from \cite{Semicomplete} that {\sc $\H$-Colouring} is \NP-hard, thus the same can be said for both {\sc $\H$-Retraction} and {\sc Surjective $\H$-Colouring}.

When $\H$ is the partially reflexive tree of Figure~\ref{fig:3digraphs} (b), {\sc $\H$-Retraction} is known to be \NP-hard from \cite{pseudoforests} and {\sc Surjective $\H$-Colouring} is known to be \NP-hard from \cite{GolovachPS12}.

When $\H$ is the reflexive tournament from Figure~\ref{fig:3digraphs} (e), the result follows from Corollary~\ref{c-main} or Lemma~\ref{prop:main1}. For the related case (f), we appeal to Lemma~\ref{lem:case-f} (finishing via Corollary~\ref{c-main}).

When $\H$ is as depicted in Figure~\ref{fig:3digraphs} (g), membership in \Ptime\ for {\sc $\H$-Retraction}, and therefore {\sc Surjective $\H$-Colouring}, follows from Lemma~\ref{lem:case-g}, in light of the recent groundbreaking proofs of \cite{FVproofBulatov,FVproofZhuk}.
\end{proof}

\section{Conclusion}\label{s-con}

We have given the first significant classification results for {\sc Surjective $\H$-Colouring} where~$\H$ comes from a class of digraphs (that are not graphs). 
To do this, we have developed both a novel algebraic method and a novel recursive combinatorial method. Below we discuss some directions for future research.

%The case of {\sc Surjective $\DC^*_3$-Colouring} is the first application of Theorem~\ref{thm:Chen-hauptsatz} to settle a problem of open complexity. 
%whose complexity remains open since it arose (under a different name) in \cite{ColoringMixedHypertrees}, see also Question~3 in~\cite{SurHomSurvey}.
Let $3\mathrm{NRC}$ be the hypergraph with vertex-set $\{r,g,b\}$ and hyperedge-set $\{r,g,b\}\setminus\{(r,g,b),(r,b,g),(g,b,r),(g,r,b),(b,r,g),(b,g,r)\}$. Then {\sc $3$-No-Rainbow-Colouring} is the problem {\sc Surjective $3\mathrm{NRC}$-Colouring}, in which one looks for a surjective colouring of the vertices, such that no hyperedge is rainbow-coloured (\mbox{i.e.} uses all colours). 
We recall that the complexity of this problem is open since it arose (under a different name) in \cite{ColoringMixedHypertrees}, see also Question~3 in~\cite{SurHomSurvey}.
The {\sc Surjective $\DC^*_3$-Colouring} problem is the digraph problem most closely related to {\sc $3$-No-Rainbow-Colouring}. To explain this, 
when looking for digraphs with a similar character to $3\mathrm{NRC}$, we would insist at least that the automorphism group is transitive. This leaves just the reflexive and irreflexive directed $3$-cycles and the reflexive and irreflexive 3-cliques, that is, $3$-cycles with a double edge between every pair of vertices (admittedly, the cycles have only some of the automorphisms of the cliques). 
If~$\H$ is the reflexive $3$-clique, then $\H$-{\sc Retraction} and {\sc Surjective $\H$-Colouring} are trivial. 
If~$\H$ is the irreflexive directed $3$-cycle, then $\H$ has a majority polymorphism, which shows that $\H$-{\sc Retraction}, and thus {\sc Surjective $\H$-Colouring} (see Figure~\ref{fig:Anthony}),
can be solved in polynomial time~\cite{Semicomplete}. 
If~$\H$ is the irreflexive $3$-clique, then {\sc Surjective $\H$-Colouring} is \NP-complete, as there exists a straightforward reduction from 3-{\sc Colouring}.
Hence~$\H=\DC^*_3$ was indeed the only case for which determining the complexity of {\sc Surjective $\H$-Colouring} was not immediately obvious. 

It would be great to extend our results to larger reflexive digraph classes. Reflexive digraphs with a double edge are not endo-trivial and further fail to be endo-trivial in the worse way, since {\sc Surjective $\DC^*_2$-Colouring} is nearly trivial. Thus, our methods are likely only to be applicable to reflexive oriented digraphs, that is, those without a double edge. On the way, a natural question arising is exactly which reflexive digraphs are endo-trivial?

Finally, there is the question as to whether the assumption of endo-triviality can be weakened to that of retract-triviality in Theorem~\ref{cor:Benoit}. Endo-triviality is used right at the beginning of the proof to show that $\G^\G$ is the disjoint union of a copy of $\G$ (the constant maps) and isolated automorphisms. We do not know if retract-triviality is here sufficient.

\subsection*{Acknowledgements} We thank several reviewers for useful comments on the first draft of the extended abstract of this paper, as well as Jan Bok for fruitful discussions.

\end{document}